\documentclass[centertags]{amsart}

\usepackage{amsmath,amsthm,amscd,amssymb}
\usepackage{dsfont}
\usepackage{enumerate}
\usepackage{txfonts}

\newcommand{\R}{\mathds{R}}
\newcommand{\C}{\mathds{C}}
\newcommand{\N}{\mathds{N}}

\newcommand{\mrm}[1]{\mathrm{#1}}
\newcommand{\e}{\mrm{e}}
\newcommand{\co}{\colon}

\newcommand{\im}{\mrm{i}}
\newcommand{\dd}{\mrm{d}}

\let\abs=\envert
\newcommand{\aabs}[1]{\bigl\lvert #1\bigr\rvert}
\newcommand{\aaabs}[1]{\left\lvert #1\right\rvert}
\newcommand{\pd}{\partial}

\DeclareMathOperator{\re}{Re}
\DeclareMathOperator{\img}{Im}

\allowdisplaybreaks
\numberwithin{equation}{section}
\newtheorem{thm}{Theorem}[section]
\newtheorem{cor}[thm]{Corollary}
\newtheorem{lem}[thm]{Lemma}
\newtheorem{prop}[thm]{Proposition}

\theoremstyle{definition}
\newtheorem{rem}{Remark}[section]
\newtheorem*{exam}{Example}
\newtheorem*{step}{Step}

\begin{document}
\title[Series expansion for the Fourier transform of a rational function]{Series expansion 
for the Fourier transform of a rational function in three dimensions}
\author{R.~Jur\v{s}\.{e}nas}
\address{Vilnius University, Institute of Theoretical Physics and Astronomy,  
A. Go\v{s}tauto 12, Vilnius 01108, Lithuania}
\email{Rytis.Jursenas@tfai.vu.lt}
\date{\today.}
\begin{abstract}
In Rashba--Dresselhaus spin-orbit coupled systems, the calculation
of Green's function requires the knowledge of the inverse
Fourier transform of rational function $P(p)/Q(p)$, where
$P(p)$ takes the values $1$ and $p^{2}$, and where
\[
Q(p)=\left(p^{2}-\zeta\right)^{2}-
\alpha^{2}\left(p_{1}^{2}+p_{2}^{2}\right)-\beta^{2}
\]
with suitable parameters
$\alpha$, $\beta\geq0$, $\zeta\in\C$. While a two-dimensional problem,
with $p=(p_{1},p_{2})$, has been recently solved [J.~Br\"{u}ning et al,
J.~Phys. A: Math. Theor. 40 (2007)], its three-dimensional 
analogue, with $p=(p_{1},p_{2},p_{3})$, remains open.
In this paper, a hypergeometric series expansion for the triple integral is provided.
Convergence of the series dependent on the parameters is studied in
detail.
\end{abstract}
\maketitle
\tableofcontents

\newpage{}
\section{Introduction} 
The spectral analysis of ultracold atomic gases is closely
related to the calculation 
of the integral kernel (Green's function) for 
Rashba--Dresselhaus spin-orbit coupled Hamiltonian.
In momentum representation the Hamiltonian is the 
operator of multiplication by
\begin{subequations}\label{eq:Hamilton}
\begin{align}
H_{R}=&\begin{pmatrix}
p^{2}+\beta & -\alpha\left(p_{2}+\im p_{1}\right) \\
-\alpha\left(p_{2}-\im p_{1}\right) & p^{2}-\beta
\end{pmatrix}, \label{eq:Hamilton-R} \\
H_{D}=&\begin{pmatrix}
p^{2}+\beta & -\alpha\left(p_{1}-\im p_{2}\right) \\
-\alpha\left(p_{1}+\im p_{2}\right) & p^{2}-\beta
\end{pmatrix}.
\label{eq:Hamilton-D}
\end{align}
\end{subequations}
The subscript $R$ (resp. $D$) indicates Rashba
(resp. Dresselhaus) type spin-orbit interaction.
Vector $p=(p_{1},p_{2})\in\R^{2}$ or 
$p=(p_{1},p_{2},p_{3})\in\R^{3}$, depending on
the model; see eg \cite{Dalibard11} and the citation therein.
The parameter $\alpha\geq0$ has a meaning of
spin-orbit-coupling strength and $\beta\geq0$ characterizes
the magnetic Zeeman field.

The representatives $H_{R}$ and $H_{D}=UH_{R}U^{*}$,
\[
U\in\left\{(-1)^{n}\e^{\mrm{i}\delta}
\begin{pmatrix}
1 & 0 \\ 0 & \im
\end{pmatrix}\co \delta\geq0,n\in\N_{0}\right\}
\subset U(2),
\] 
of the Hamiltonian (\ref{eq:Hamilton}) are unitarily equivalent.
The Green's function 
\[
\mathcal{G}_{R}(x)\equiv 
\mathcal{G}_{R}(x;\alpha,\beta,\zeta),
\quad
x=(x_{1},x_{2},x_{3})\in\R^{3}-\{0\}
\]
for $H_{R}$ is 
\begin{equation}
\mathcal{G}_{R}(x)=
\begin{pmatrix}
G_{2}(x)-\beta G_{1}(x) &
-\alpha D_{-}G_{1}(x) \\
\alpha D_{+}G_{1}(x) & G_{2}(x)+\beta G_{1}(x)
\end{pmatrix}
\label{eq:Green}
\end{equation}
where 
\begin{subequations}\label{eq:Integral}
\begin{align}
G_{1}(x)\equiv& G_{1}(x;\alpha,\beta,\zeta)=
\int_{\R^{3}}\frac{\dd^{3}p}{(2\pi)^{3}}\e^{-\im p.x}
\frac{ 1 }{ Q(p) }, \label{eq:Integral-1} \\
G_{2}(x)\equiv& G_{2}(x;\alpha,\beta,\zeta)=
\int_{\R^{3}}\frac{\dd^{3}p}{(2\pi)^{3}}\e^{-\im p.x}
\frac{ p^{2}-\zeta }{ Q(p) } \label{eq:Integral-2}
\intertext{where}
D_{\pm}=&\frac{\pd}{\pd x_{1}}\pm\im 
\frac{\pd}{\pd x_{2}}, \quad 
Q(p)=\left(p^{2}-\zeta\right)^{2}-
\alpha^{2}\left(p_{1}^{2}+p_{2}^{2}\right)-\beta^{2}
\nonumber
\end{align}
\end{subequations}
where we have used 
\[
\im\frac{\pd}{\pd x_{j}}G_{1}(x)=
\int_{\R^{3}} \frac{\dd^{3}p}{(2\pi)^{3}}\e^{-\im p.x}
\frac{p_{j}}{Q(p)},\quad  j=1,2.
\]
The complex number $\zeta\in\C$ is in the
resolvent set of $H$ ($=H_{R},H_{D}$). When considered
separately, $G_{1}$ is defined for $x\in\R^{3}$, and
$G_{2}$ for $x\in\R^{3}-\{0\}$. 
The associated Green's function for $H_{D}$ is
$U\mathcal{G}_{R}U^{*}$.

The resolvent set of $H$ consists of complex numbers
$\zeta$ for which the denominator $Q(p)$, $\forall p\in\R^{3}$,
is nonzero:
\begin{equation}
\zeta\in\C-\left[-\Sigma,\infty\right),\quad
\Sigma=\begin{cases}
\beta, & \text{if}\quad \beta>\frac{1}{2}\alpha^{2}, \\
\displaystyle\left(\frac{\beta}{\alpha}\right)^{2}+
\left(\frac{\alpha}{2}\right)^{2}, & 
\text{if}\quad \beta\leq\frac{1}{2}\alpha^{2}.
\end{cases}
\label{eq:zeta}
\end{equation}
Stated otherwise, $\left[-\Sigma,\infty\right)$ is
the essential spectrum of $H$.

In this paper, our principle goal is to examine
(\ref{eq:Green})--(\ref{eq:Integral}) with $\alpha,\beta\geq0$
(due to physical reasons)
and $\zeta$ as in (\ref{eq:zeta}) (to ensure the existence of
$(H-\zeta)^{-1}$).

While for the spinless systems, that is, for $\alpha=0$, the integrals
(\ref{eq:Integral}) are easy to compute for suitable $x$, $\alpha$, $\beta$,
$\zeta$ [rewrite $p$ in spherical
coordinates, integrate over the angles and then apply 
\cite[Eq.~(3.728.2)]{Gradshteyn07}], 
the case when $\alpha>0$ is much more involved.

No evidence of successful algebraic treatment of 
the integrals (\ref{eq:Integral}) has been found so far.
The special case when $\alpha=0$ has been discussed in
\cite{Cacciapuoti09}.
A two dimensional equivalent of (\ref{eq:Integral}), 
with $p=(p_{1},p_{2})\in \R^{2}$ and $x\in\R^{2}$, 
has been computed in
\cite{Bruning07,Exner07,Carlone11}. In fact,
the derivation of the Green's function in dimension two
does not require an explicit calculation of (\ref{eq:Integral}), for
one explores the fact that the square of spin-orbit term 
in the Hamiltonian is just the two-dimensional Laplace operator. 
In other words, the equality $p^{2}=p_{1}^{2}+p_{2}^{2}$ is
critical in dimension two. This is no longer the case in
dimension three, and thus the problem stands in need of
different computational methods.

In this paper, a hypergeometric series expansion for 
(\ref{eq:Green})--(\ref{eq:zeta}) is provided
(\S\S \ref{sec:3}--\ref{sec:4},\ref{sec:7}).
Convergence conditions on $\alpha$, $\beta$, $\zeta$
are studied in detail. Some properties and special cases
of the series are also discussed (\S\S \ref{sec:4},\ref{sec:6}--\ref{sec:7}).
\section{Notation and terminology}
The hypergeometric series to be used are the
Kamp\'{e} de F\'{e}riet function $F_{l:m;n}^{p:q;k}$,
the generalized Lauricella function of two variables
$F_{C:D;D^{\prime}}^{A:
B;B^{\prime}}$, the complete and confluent
Horn, Appell, Humbert functions.

The Kamp\'{e} de F\'{e}riet function $F_{l:m;n}^{p:q;k}$
is defined according to \cite[\S 1.7, Eq.~(16)]{Srivastava84}
\begin{align*}
F_{l:m;n}^{p:q;k}&\left(\begin{array}{rrr}
a_{1},\ldots,a_{p}: & b_{1},\ldots,b_{q}; & c_{1},\ldots,c_{k}; \\
\alpha_{1},\ldots,\alpha_{l}: & \beta_{1},\ldots,\beta_{m}; & 
\gamma_{1},\ldots,\gamma_{n};
\end{array}\zeta_{1},\zeta_{2}\right) \\
&=\sum_{r,s=0}^{\infty}
\frac{ \prod_{j=1}^{p}(a_{j})_{r+s}\prod_{j=1}^{q}(b_{j})_{r}
\prod_{j=1}^{k}(c_{j})_{s} }{ \prod_{j=1}^{l}(\alpha_{j})_{r+s}
\prod_{j=1}^{m}(\beta_{j})_{r}\prod_{j=1}^{n}(\gamma_{j})_{s} }
\frac{\zeta_{1}^{r}\zeta_{2}^{s}}{r!s!}
\end{align*}
where the parentheses indicate the Pochhammer symbol. 
In the present paper, the Kamp\'{e} de F\'{e}riet functions to be used
fulfill $p+q<l+m+1$, $p+k<l+n+1$. This ensures the convergence
for all $\abs{\zeta_{1}},\abs{\zeta_{2}}<\infty$.

The generalized Lauricella function of two variables
$F_{C:D;D^{\prime}}^{A:B;B^{\prime}}$ is due to
Srivastava--Daoust \cite[\S 1.7, Eq.~(18)]{Srivastava84},
\begin{align*}
F_{C:D;D^{\prime}}^{A:B;B^{\prime}}&\left(\begin{array}{rrr}
\left((a):\theta,\theta^{\prime}\right): 
& \left((b):\phi\right); & \left((b^{\prime}):\phi^{\prime}\right); \\
\left((c):\psi,\psi^{\prime}\right): 
& \left((d):\delta\right); & \left((d^{\prime}):\delta^{\prime}\right);
\end{array}\zeta_{1},\zeta_{2}\right) \\
&=\sum_{r,s=0}^{\infty}
\frac{ \prod_{j=1}^{A}(a_{j})_{r\theta_{j}+s\theta_{j}^{\prime}}
\prod_{j=1}^{B}(b_{j})_{r\phi_{j}}
\prod_{j=1}^{B^{\prime}}(b_{j}^{\prime})_{s\phi_{j}^{\prime}} }{ 
\prod_{j=1}^{C}(c_{j})_{r\psi_{j}+s\psi_{j}^{\prime}}
\prod_{j=1}^{D}(d_{j})_{r\delta_{j}}
\prod_{j=1}^{D^{\prime}}(d_{j}^{\prime})_{s\delta_{j}^{\prime}} }
\frac{\zeta_{1}^{r}\zeta_{2}^{s}}{r!s!}
\end{align*}
where the coefficients 
\[
\left(\theta_{j}\co j=1,\ldots,A\right),\quad
\left(\theta_{j}^{\prime}\co j=1,\ldots,A\right),\ldots,
\left(\delta_{j}^{\prime}\co j=1,\ldots,D^{\prime}\right)
\]
are real and positive, and $(a)$ abbreviates the array of $A$ parameters
$a_{1},\ldots,a_{A}$, with similar interpretation for $(b),\ldots,
(d^{\prime})$. In the present paper, the coefficients are such that
the convergence is ensured for all 
$\abs{\zeta_{1}},\abs{\zeta_{2}}<\infty$. For other cases, one is referred to
\cite{Srivastava72,Srivastava92}.

The complete Horn $H_{3}$ function \cite[\S 5.7.1, Eq.~(15)]{Erdelyi53}, the 
confluent Horn $\mrm{H}_{3}$ and $\mrm{H}_{10}$ functions
\cite[\S 5.7.1, Eqs.~(31) and (38)]{Erdelyi53}, the confluent Appell
(or Humbert) $\Xi_{2}$ function
\cite[\S 1.6, Eq.~(44)]{Srivastava84} obey the following
series representation
\begin{align*}
H_{3}(a,b;c;\zeta_{1},\zeta_{2})=&\sum_{m,n=0}^{\infty}
\frac{ (a)_{2m+n}(b)_{n} }{ (c)_{m+n} }
\frac{\zeta_{1}^{m}}{m!}\frac{\zeta_{2}^{n}}{n!},
\quad \abs{\zeta_{1}}<R,\quad \abs{\zeta_{2}}<S, \\
&R+\left(S-\frac{1}{2}\right)^{2}=\frac{1}{4}, \\
\mrm{H}_{3}(a,b;c;\zeta_{1},\zeta_{2})=&\sum_{m,n=0}^{\infty}
\frac{ (a)_{m-n}(b)_{m} }{ (c)_{m} }
\frac{\zeta_{1}^{m}}{m!}\frac{\zeta_{2}^{n}}{n!},\quad
\abs{\zeta_{1}}<1,\quad\abs{\zeta_{2}}<\infty, \\
\mrm{H}_{10}(a;c;\zeta_{1},\zeta_{2})=&\sum_{m,n=0}^{\infty}
\frac{ (a)_{2m-n} }{ (c)_{m} }
\frac{\zeta_{1}^{m}}{m!}\frac{\zeta_{2}^{n}}{n!},\quad
\abs{\zeta_{1}}<\frac{1}{4},\quad\abs{\zeta_{2}}<\infty, \\
\Xi_{2}(a,b;c;\zeta_{1},\zeta_{2})=&\sum_{m,n=0}^{\infty}
\frac{ (a)_{m}(b)_{m} }{ (c)_{m+n} }
\frac{\zeta_{1}^{m}}{m!}\frac{\zeta_{2}^{n}}{n!},\quad
\abs{\zeta_{1}}<1,\quad\abs{\zeta_{2}}<\infty,
\end{align*}
for suitable parameters $a$, $b$, $c\in\C$.

Various properties
of the Appell and Horn confluent functions are derived in
\cite{Srivastava84,Erdelyi53,Kalnins80}. A systematic and modern
treatment of these functions is in \cite{Debiard02,Debiard03}.

Throughout the paper the absence of parameters
in either series is left blank and in that case the value of an
empty product is unity.
\section{Main results}\label{sec:3}
\begin{lem}\label{lem:X}
Let $\zeta=(\zeta_{1},\zeta_{2},\zeta_{3})\in\C^{3}$
and define the triple series
\begin{equation}
X(a,b;\zeta)=\sum_{m,n,p=0}^{\infty}
\frac{ \zeta_{1}^{m}\zeta_{2}^{n}\zeta_{3}^{p} }{ 
m!p!(a)_{2m+n+p}(b)_{m+n} }, 
\label{eq:X}
\end{equation}
$\forall a,b\in\C-\{-n\co n\in\N_{0}\}$. Then, the series 
(\ref{eq:X}) takes the following equivalent representations
\begin{subequations}\label{eq:X123}
\begin{align}
X(a,b;\zeta)=&
\sum_{n=0}^{\infty}
\frac{ \zeta_{1}^{n} }{ n!(a)_{2n}(b)_{n} }F_{1:1;0}^{0:1;0}\left(
\begin{array}{rr}
: & 1;; \\
a+2n: & b+n;;
\end{array}\zeta_{2},\zeta_{3}\right) \label{eq:X1} \\
=&
\sum_{n=0}^{\infty}
\frac{ \zeta_{2}^{n} }{ (a)_{n}(b)_{n} }F_{1:1;0}^{0:0;0}\left(
\begin{array}{rr}
: & ;; \\
\left(a+n:2,1\right): & \left(b+n:1\right);;
\end{array}\zeta_{1},\zeta_{3}\right) \label{eq:X2} \\
=&
\sum_{n=0}^{\infty}
\frac{ \zeta_{3}^{n} }{ n!(a)_{n} }F_{2:0;0}^{0:0;1}\left(
\begin{array}{rrr}
:\, ; & \left(1:1\right); \\
\left(a+n:2,1\right),\left(b:1,1\right):\, ;& ;
\end{array}\zeta_{1},\zeta_{2}\right) \label{eq:X3}
\end{align}
\end{subequations}
$\forall\abs{\zeta}<\infty$. Moreover, $X$ fulfills the recurrence relation
\begin{equation}
X(a,b;\zeta)-
\frac{\zeta_{2}}{ab}X(a+1,b+1;\zeta)=
F_{1:1;0}^{0:0;0}\left(
\begin{array}{rr}
: & ;; \\
\left(a:2,1\right): & \left(b:1\right);;
\end{array}\zeta_{1},\zeta_{3}\right).
\label{eq:recurrX}
\end{equation}
\end{lem}
\begin{lem}\label{lem:Xp}
Let $\zeta=(\zeta_{1},\zeta_{2},\zeta_{3})\in\C^{3}$ and define 
the triple series
\begin{equation}
X^{\prime}(a,b;\zeta)=\sum_{m,n,p=0}^{\infty}
\frac{ \zeta_{1}^{m}\zeta_{2}^{n}\zeta_{3}^{p}
(a)_{2m-n-p} }{ (m-n)!p!(b)_{m} }, 
\label{eq:Xp}
\end{equation}
$\forall a\in\C-\N$, $\forall b\in \C-\{-n\co n\in\N_{0}\}$.
Then, the series (\ref{eq:Xp}) is absolutely convergent if:
\begin{subequations}\label{eq:Xp123}
\begin{align}
X^{\prime}(a,b;\zeta)=&
\sum_{n=0}^{\infty}\frac{ \zeta_{1}^{n}(a)_{2n} }{ n!(b)_{n} }
\Xi_{2}(1,-n;1-a-2n;\zeta_{2},-\zeta_{3}), \nonumber \\
&\abs{\zeta_{1}}<\frac{1}{4},\quad \abs{\zeta_{2}}<2
\quad\text{or} \nonumber \\
&\abs{\zeta_{1}}=\frac{1}{4},\quad \abs{\zeta_{2}}<2,\quad
\re\left(a-b-\frac{1}{2}\right)<0 \label{eq:Xp1} \\
=&
\sum_{n=0}^{\infty}\frac{ (\zeta_{1}\zeta_{2})^{n}
(a)_{n} }{ (b)_{n} }
\mrm{H}_{10}(a+n;b+n;\zeta_{1},\zeta_{3}), \nonumber \\
&\abs{\zeta_{1}}<\frac{1}{4},\quad
\abs{\zeta_{2}}<\frac{ 1+\sqrt{1-4\abs{\zeta_{1}}} }{ 
2\abs{\zeta_{1}} }\quad
\text{or} \nonumber \\
&\abs{\zeta_{1}}<\frac{1}{4},\quad
\abs{\zeta_{2}}=\frac{ 1+\sqrt{1-4\abs{\zeta_{1}}} }{ 
2\abs{\zeta_{1}} },\quad
\re(a-b)<0 \label{eq:Xp2} \\
=&
\sum_{n=0}^{\infty}\frac{ (-\zeta_{3})^{n} }{ n!(1-a)_{n} }
H_{3}(a-n,1;b;\zeta_{1},\zeta_{1}\zeta_{2}), \nonumber \\
&\abs{\zeta_{1}}<R,\quad\abs{\zeta_{1}\zeta_{2}}<S,\quad
R+\left(S-\frac{1}{2}\right)^{2}=\frac{1}{4}, \label{eq:Xp3}
\end{align}
\end{subequations}
$\forall\abs{\zeta_{3}}<\infty$. 
Moreover, $X^{\prime}$ fulfills the recurrence relation
\begin{equation}
X^{\prime}(a,b;\zeta)-\frac{a}{b}\zeta_{1}\zeta_{2}
X^{\prime}(a+1,b+1;\zeta)=
\mrm{H}_{10}(a;b;\zeta_{1},\zeta_{3}),
\label{eq:recurrXp}
\end{equation}
$\abs{\zeta_{1}}<\frac{1}{4}$. The confluence on a pair
$(\zeta_{1},\zeta_{2})$ implies that
\begin{equation}
\lim_{\epsilon\to0}
X^{\prime}\left(a,b;\epsilon\zeta_{1},
\frac{\zeta_{2}}{\epsilon},\zeta_{3}\right)=
\mrm{H}_{3}(a,1;b;\zeta_{1}\zeta_{2},\zeta_{3}),
\label{eq:conflXp}
\end{equation} 
$\abs{\zeta_{1}\zeta_{2}}<1$.
\end{lem}
\begin{thm}\label{thm:Integral}
Let $\alpha$, $\beta\geq0$, $p=(p_{1},p_{2},p_{3})\in\R^{3}$,
$x\in\R^{3}$, $r=\abs{x}$. Suppose that
$\C-\left[-\Sigma,\infty\right)\ni\zeta$ meets 
at least one set of the following three:
\begin{enumerate}[\upshape ($a$)]
\item $2\beta>\alpha^{2}$,
\quad $\displaystyle\beta\leq\abs{\zeta}<2
\left(\frac{\beta}{\alpha}\right)^{2}$
\item $\abs{\zeta}\geq\Sigma$; the equality is available
only if $0\leq2\beta<\alpha^{2}$
\item $\displaystyle\abs{\zeta}>\max\left(\frac{\beta}{2\sqrt{R}},
\frac{\alpha^{2}}{4S}\right)$, \quad 
$\displaystyle R+\left(S-\frac{1}{2}\right)^{2}=\frac{1}{4}$.
\end{enumerate}
Then 
\begin{align}
G_{1}(x;\alpha,\beta,\zeta)=&
\frac{1}{ 8\pi\sqrt{-\zeta} }X^{\prime}\left(
\frac{1}{2},\frac{3}{2};\frac{\beta^{2}}{ 4\zeta^{2} },
-\frac{\zeta\alpha^{2}}{ \beta^{2} },
\frac{ \zeta r^{2} }{4}\right) \nonumber \\
&-\frac{r}{ 8\pi }X\left(
\frac{3}{2},\frac{3}{2};\frac{\beta^{2}r^{4}}{ 64 },
-\frac{\alpha^{2}r^{2}}{ 16 },
-\frac{ \zeta r^{2} }{4}\right),\quad r\geq0,
\label{eq:Integral1} \\
G_{2}(x;\alpha,\beta,\zeta)=&
\frac{1}{4\pi r}X\left(
\frac{1}{2},\frac{1}{2};\frac{\beta^{2}r^{4}}{ 64 },
-\frac{\alpha^{2}r^{2}}{ 16 },
-\frac{ \zeta r^{2} }{4}\right) \nonumber \\
&-\frac{\sqrt{-\zeta}}{4\pi}X^{\prime}\left(
-\frac{1}{2},\frac{1}{2};\frac{\beta^{2}}{ 4\zeta^{2} },
-\frac{\zeta\alpha^{2}}{ \beta^{2} },
\frac{ \zeta r^{2} }{4}\right),\quad r>0.
\label{eq:Integral2}
\end{align}
Conditions (a), (b) and (c) indicate that $X^{\prime}$ is given by
(\ref{eq:Xp1}), (\ref{eq:Xp2}) and (\ref{eq:Xp3}), respectively.
The series representation for $X$ admits any form
given in (\ref{eq:X123}).
\end{thm}
\begin{rem}\label{rem:zetadomain}
The series expansion (\ref{eq:Integral1})--(\ref{eq:Integral2})
for the triple integrals
(\ref{eq:Integral}) remains valid
for complex $\alpha$, $\beta$: In this case, 
parameters $\alpha$, $\beta$ in ($a$)--($c$) must be
replaced by their corresponding absolute values and
the domain (\ref{eq:zeta}) would change due to
the solutions to $Q(p)=0$, $\forall p\in\R^{3}$; these 
have to be excluded from the whole plane $\C$.
\end{rem}
\section{Special cases}\label{sec:4}
For particular values of parameters,
the series representation (\ref{eq:Integral1})--(\ref{eq:Integral2})
of integrals (\ref{eq:Integral}) can be considerably simplified.
For illustrative purposes, we examine the
limits $\alpha\to0$, $\beta\to0$, and $r\to0$.

1. In the limit $\alpha\to0$,
\begin{align}
G_{1}(x;0,\beta,\zeta)=&
\frac{1}{8\pi\sqrt{-\zeta}}
\mrm{H}_{10}\left(\frac{1}{2};\frac{3}{2};
\frac{\beta^{2}}{4\zeta^{2}},\frac{\zeta r^{2}}{4}\right)
\nonumber \\
&-\frac{r}{8\pi}\,
F_{1:1;0}^{0:0;0}\left(
\begin{array}{rr}
: & ;; \\
\left(\frac{3}{2}:2,1\right): & \left(\frac{3}{2}:1\right);;
\end{array}\frac{\beta^{2}r^{4}}{64},
-\frac{\zeta r^{2}}{4}\right)
\label{eq:G1alpha0}
\end{align}
$\forall\beta\geq0$, $\forall\zeta\in\C-
\left[-\beta,\infty\right)$, $\abs{\zeta}>\beta$,
$\forall r\geq0$.
In comparison, a direct calculation of (\ref{eq:Integral-1})
with $\alpha=0$ gives (see also \cite{Cacciapuoti09})
\begin{equation}
G_{1}(x;0,\beta,\zeta)=
\frac{1}{8\pi \beta r}
\left(\e^{ -r\sqrt{-\beta-\zeta} }
-\e^{ -r\sqrt{\beta-\zeta} }\right)
\label{eq:G1alpha0-2}
\end{equation}
$\forall\beta\geq0$, $\forall\zeta\in\C-
\left[-\beta,\infty\right)$, $\forall r\geq0$.
Equation~(\ref{eq:G1alpha0-2}) requires a milder
condition on $\zeta$ than (\ref{eq:G1alpha0}).
Define, for convenience,
\[
\zeta_{1}=\frac{\beta^{2}}{4\zeta^{2}},\quad
\zeta_{2}=\frac{\zeta r^{2}}{4}.
\]
It is not hard to see from elementary algebra that, 
$\forall\abs{\zeta_{1}}<\frac{1}{4}$,
\begin{align}
\mrm{H}_{10}\left(\frac{1}{2};\frac{3}{2};
\zeta_{1},\zeta_{2}\right)=
\sum_{\sigma=\pm1}\frac{ \sigma\sqrt{1+2\sigma 
\sqrt{\zeta_{1}} } }{ 2\sqrt{\zeta_{1}} } 
\,_{0}F_{1}\left(;\frac{3}{2};-\zeta_{2}
\left(1+2\sigma\sqrt{\zeta_{1}}\right)\right).
\label{eq:H10}
\end{align}
The function $_{0}F_{1}$ is entire. If one substituted
the right-hand side of (\ref{eq:H10}) in
(\ref{eq:G1alpha0}), the condition
$\abs{\zeta}>\beta$ would be removed by analytic continuation.

Similar considerations apply to $G_{2}$
at $\alpha=0$. By Theorem~\ref{thm:Integral},
\begin{align}
G_{2}(x;0,\beta,\zeta)=&
\frac{1}{4\pi r}\,F_{1:1;0}^{0:0;0}\left(
\begin{array}{rr}
: & ;; \\
\left(\frac{1}{2}:2,1\right): & \left(\frac{1}{2}:1\right);;
\end{array}\frac{\beta^{2}r^{4}}{64},
-\frac{\zeta r^{2}}{4}\right) \nonumber \\
&-\frac{ \sqrt{-\zeta} }{ 4\pi }
\mrm{H}_{10}\left(-\frac{1}{2};\frac{1}{2};
\frac{\beta^{2}}{4\zeta^{2}},\frac{\zeta r^{2}}{4}\right)
\label{eq:G2alpha0}
\end{align}
$\forall\beta\geq0$, $\forall\zeta\in\C-
\left[-\beta,\infty\right)$, $\abs{\zeta}>\beta$,
$\forall r>0$.
A direct calculation of (\ref{eq:Integral-2})
with $\alpha=0$ gives
\begin{equation}
G_{2}(x;0,\beta,\zeta)=
\frac{1}{8\pi r}
\left(\mrm{e}^{ -r\sqrt{-\beta-\zeta} }
+\mrm{e}^{ -r\sqrt{\beta-\zeta} }\right)
\label{eq:G2alpha0-2}
\end{equation}
$\forall\beta\geq0$, $\forall\zeta\in\C-
\left[-\beta,\infty\right)$, $\forall r>0$.
Again, in view of
\[
\mrm{H}_{10}\left(-\frac{1}{2};\frac{1}{2};
\zeta_{1},\zeta_{2}\right)=
\sum_{\sigma=\pm1}
\frac{ \sqrt{ 1+2\sigma\sqrt{\zeta_{1}} } }{ 2 }\,
_{0}F_{1}\left(;\frac{3}{2};-\zeta_{2}\left(
1+2\sigma\sqrt{\zeta_{1}}\right)\right)
\]
$\forall\abs{\zeta_{1}}<\frac{1}{4}$,
the additional condition $\abs{\zeta}>\beta$ in
(\ref{eq:G2alpha0}) can be relaxed.

2. In the limit $\beta\to0$, $\forall\alpha\geq0$,
$\forall\zeta\in \C-\left[-\frac{1}{4}\alpha^{2},\infty\right)$,
$\abs{\zeta}>\frac{1}{4}\alpha^{2}$,
\begin{align}
G_{1}(x;\alpha,0,\zeta)=&
\frac{1}{ 8\pi\sqrt{-\zeta} }
\mrm{H}_{3}\left(\frac{1}{2},1;\frac{3}{2};
-\frac{\alpha^{2}}{4\zeta},\frac{\zeta r^{2}}{4}
\right) \nonumber \\
&-\frac{r}{8\pi}\,
F_{1:1;0}^{0:1;0}\left(
\begin{array}{rr}
: & 1;; \\
\frac{3}{2}: & \frac{3}{2};;
\end{array}-\frac{\alpha^{2}r^{2}}{16},
-\frac{\zeta r^{2}}{4}\right), \quad
r\geq0, \label{eq:Spec3} \\
G_{2}(x;\alpha,0,\zeta)=&
\frac{1}{4\pi r}\,
F_{1:1;0}^{0:1;0}\left(
\begin{array}{rr}
: & 1;; \\
\frac{1}{2}: & \frac{1}{2};;
\end{array}-\frac{\alpha^{2}r^{2}}{16},
-\frac{\zeta r^{2}}{4}\right) \nonumber \\
&-\frac{\sqrt{-\zeta}}{4\pi}
\mrm{H}_{3}\left(-\frac{1}{2},1;\frac{1}{2};
-\frac{\alpha^{2}}{4\zeta},\frac{\zeta r^{2}}{4}
\right),\quad
r>0. \label{eq:Spec4}
\end{align}
Taking both $\alpha=\beta=0$, we find from
(\ref{eq:G1alpha0})--(\ref{eq:Spec4}) that
\[
G_{1}(x;0,0,\zeta)=\frac{\e^{ -r\sqrt{-\zeta} }}{ 
8\pi\sqrt{-\zeta} }\quad (r\geq0),\quad
G_{2}(x;0,0,\zeta)=\frac{\e^{ -r\sqrt{-\zeta} }}{ 
4\pi r }\quad (r>0)
\]
$\forall\zeta\in\C-\left[0,\infty\right)$. As is seen,
function $G_{2}(x;0,0,\zeta)$ is the Green's function
for the three-dimensional kinetic energy operator
(recall (\ref{eq:Hamilton})).

3. In physical applications, the limit $r\to0$ in Green's
function is usually
necessary for the calculation of point spectrum of
the operator perturbed by the point-interaction
\cite{Albeverio13,Albeverio10,Albeverio00,Malamud12}.

By the theorem,
\begin{equation}
G_{1}(0;\alpha,\beta,\zeta)=\frac{1}{ 8\pi\sqrt{-\zeta} }
H_{3}\left(\frac{1}{2},1;\frac{3}{2};
\frac{\beta^{2}}{4\zeta^{2}},-\frac{\alpha^{2}}{4\zeta}
\right)
\label{eq:Spec1}
\end{equation}
where $\zeta$ meets (\ref{eq:zeta}) and ($c$)
in Theorem~\ref{thm:Integral}. An additional condition~($c$)
ensuring the convergence of the complete Horn $H_{3}$
can be relaxed on account of
\[
H_{3}\left(\frac{1}{2},1;\frac{3}{2};
\frac{\beta^{2}}{4\zeta^{2}},-\frac{\alpha^{2}}{4\zeta}
\right)=
\frac{ 2\sqrt{-\zeta} }{ \alpha }\mrm{artanh}\left(
\frac{\alpha}{\beta}
\sqrt{ \frac{-\zeta}{2}
\left(1-\sqrt{ 1-\left(\frac{ \beta }{ \zeta } \right)^{2} }\right) }\right).
\]
The inverse hyperbolic tangent $\mrm{artanh}$
is defined on $\C-\{-1,1\}$; the latter
is always satisfied for $\zeta$ as in (\ref{eq:zeta}).

Function $G_{2}$ exists
$\forall x\in\R^{3}-\{0\}$, but its renormalized counterpart
\begin{equation}
G_{2}^{\mrm{ren}}(x;\alpha,\beta,\zeta)=
G_{2}(x;\alpha,\beta,\zeta)-
\frac{\e^{ -r\sqrt{-\zeta} }}{ 4\pi r }
\label{eq:Spec2a}
\end{equation}
exists $\forall x\in\R^{3}$. Indeed, by the theorem,
\begin{equation}
G_{2}^{\mrm{ren}}(0;\alpha,\beta,\zeta)=
\frac{\sqrt{-\zeta}}{4\pi}
-\frac{\sqrt{-\zeta}}{4\pi}
H_{3}\left(-\frac{1}{2},1;\frac{1}{2};
\frac{\beta^{2}}{4\zeta^{2}},-\frac{\alpha^{2}}{4\zeta}
\right)
\label{eq:Spec2}
\end{equation}
provided that $\zeta$ is as in (\ref{eq:zeta}) and 
Theorem~\ref{thm:Integral}-($c$). Likewise,
condition~($c$) can be omitted if noting that
\begin{align*}
H_{3}\left(-\frac{1}{2},1;\frac{1}{2};
\frac{\beta^{2}}{4\zeta^{2}},-\frac{\alpha^{2}}{4\zeta}
\right)=&
\sqrt{ \frac{1}{2}
\left(1+\sqrt{ 1-\left(\frac{ \beta }{ \zeta } \right)^{2} }\right) } \\
&-\frac{ \alpha }{ 2\sqrt{-\zeta} }\mrm{artanh}\left(
\frac{\alpha}{\beta}
\sqrt{ \frac{-\zeta}{2}
\left(1-\sqrt{ 1-\left(\frac{ \beta }{ \zeta } \right)^{2} }\right) }\right).
\end{align*}
The function $G_{2}^{\mrm{ren}}$ 
(\ref{eq:Spec2a})--(\ref{eq:Spec2}) appears explicitly in
the theory of singular perturbations \cite{Albeverio00}
when dealing with self-adjoint extensions of operators with
point-interaction; see also \S \ref{sec:8}.
\section{Demonstration of main results}
The proof of Lemma~\ref{lem:X} is straightforward and thus omitted: 
It requires nothing more than the definition of the Kamp\'{e} de F\'{e}riet and Lauricella
functions and elementary rearrangement of summands due to Pochhammer symbol.

The proof of Lemma~\ref{lem:Xp} is more involved.
\begin{proof}
\begin{step}[Series representation]
To prove (\ref{eq:Xp1}), substitute
\[
(a)_{2m-n-p}=\frac{ (-1)^{n+p}(a)_{2m} }{ 
(1-a-2m)_{n+p} },\quad
(m-n)!=(1)_{m-n}=\frac{ (-1)^{n}(1)_{m} }{ (-m)_{n} }
\]
in (\ref{eq:Xp}) and apply the series representation of $\Xi_{2}$.

To prove (\ref{eq:Xp2}), let $m=n+p+q$. Then $q=-n-p,-n-p+1,\ldots$ 
But then
\[
\frac{1}{ (1)_{m-n} }=\frac{1}{ (1)_{p+q} }=0\quad
\text{for}\quad p+q=-1,-2,\ldots
\]
Turns out that $q=-p,-p+1,\ldots$ and, by (\ref{eq:Xp}),
\[
X^{\prime}(a,b;\zeta)=
\sum_{n,p=0}^{\infty}\sum_{q=-p}^{\infty}
\zeta_{1}^{q}\frac{ (\zeta_{1}\zeta_{2})^{n} }{ n! }
\frac{ (\zeta_{1}\zeta_{3})^{p} }{ p! }
\frac{ (1)_{n}(a)_{2q+n+p} }{ (b)_{n+p+q}(1)_{p+q} }.
\]
Let $q=l-p$. Then $l=0,1,\ldots$ and
\begin{align}
X^{\prime}(a,b;\zeta)=&
\sum_{n=0}^{\infty}(\zeta_{1}\zeta_{2})^{n}\sum_{l,p=0}^{\infty}
\frac{ \zeta_{1}^{l} }{ l! }\frac{ \zeta_{3}^{p} }{ p! }
\frac{ (a)_{2l+n-p} }{ (b)_{l+n} } \label{eq:Series} \\
=&\sum_{n=0}^{\infty}(\zeta_{1}\zeta_{2})^{n}
\frac{ (a)_{n} }{ (b)_{n} }\sum_{l,p=0}^{\infty}
\frac{ \zeta_{1}^{l} }{ l! }\frac{ \zeta_{3}^{p} }{ p! }
\frac{ (a+n)_{2l-p} }{ (b+n)_{l} }. \nonumber
\end{align}
The double sum over $l$, $p$ represents the confluent Horn $\mrm{H}_{10}$, thus
showing (\ref{eq:Xp2}).

To prove (\ref{eq:Xp3}), note that
\[
\frac{ (a)_{2l+n-p} }{ (b)_{l+n} }=
\frac{ (-1)^{p}(a-p)_{2l+n} }{ (1-a)_{p}(b)_{l+n} }.
\]
Substitute the right-hand side in (\ref{eq:Series}) and get that the
double sum over $l$, $n$ is the complete Horn $H_{3}$ as in (\ref{eq:Xp3}).
\end{step}
\begin{step}[Convergence]
The formula to be used is 
\cite[\S 2.11, Eq.~(4)]{Luke69}
\begin{equation}
\ln\Gamma(z+a)=
\left(z+a-\frac{1}{2}\right)\ln(z)-z+\frac{1}{2}\ln(2\pi)+
\sum_{k=1}^{\infty}
\frac{ (-1)^{k+1}B_{k+1}(a) }{ k(k+1) }\,z^{-k},
\label{eq:Gamma}
\end{equation}
$\aabs{\arg(z)}\leq\pi-\epsilon$, $\epsilon>0$,
$a$, $z\in\C$, where $B_{k+1}(a)$ is the Bernoulli
polynomial in $a$ of degree $k+1$;
$B_{k}(0)\equiv B_{k}$ is the Bernoulli number.

Another asymptotic formula to be used in the proof
follows from (\ref{eq:Gamma}); see also
\cite[Eq.~(1.4)]{Srivastava72}, \cite[\S 2.11, Eq.~(11)]{Luke69}:
\begin{equation}
\frac{\Gamma(a+z)}{\Gamma(b+z)}=
z^{a-b}\left(1+\frac{(a-b)(a+b-1)}{2z}+
O\left(z^{-2}\right)\right).
\label{eq:gamma}
\end{equation}

The confluent Appell $\Xi_{2}$ function in (\ref{eq:Xp1}) is a polynomial
in $\zeta_{2}$ of degree $n$. As a result, the condition $\abs{\zeta_{2}}<1$
is slightly relaxed. Indeed, by (\ref{eq:gamma}), the Pochhammer symbol
$(-n)_{m}\sim(-n)^{m}$ for $m=0,1,\ldots,n$ as $n\to\infty$. Also,
\[
\frac{1}{(1-a-2n)_{m+p}}=(-1)^{m+p}(2n)^{-m-p}
\left(1+O\left(n^{-1}\right)
\right),\quad\forall m,p\in\N_{0}.
\]
Subsequently, 
\[
\lim_{n\to\infty}\Xi_{2}(1,-n;1-a-2n;\zeta_{2},-\zeta_{3})=
\frac{1}{ 1-\zeta_{2}/2 },\quad\abs{\zeta_{2}}<2.
\]
Define
\[
a_{n}=\frac{ \zeta_{1}^{n}(a)_{2n} }{ n!(b)_{n} }
\Xi_{2}(1,-n;1-a-2n;\zeta_{2},-\zeta_{3}).
\]
Then, $\forall \abs{\zeta_{2}}<2$,
\[
a_{n}\sim \frac{ \Gamma(b) }{ \Gamma\left(\frac{a}{2}\right)
\Gamma\left(\frac{a+1}{2}\right)\left(1-\zeta_{2}/2\right) }\,
(4\zeta_{1})^{n}n^{a-b-\frac{1}{2}}\quad
\text{as}\quad n\to\infty.
\]
But then, $\lim a_{n}=0$ if either $\abs{\zeta_{1}}<\frac{1}{4}$ or
$\abs{\zeta_{1}}=\frac{1}{4}$, $\re\left(a-b-\frac{1}{2}\right)<0$.
The ratio test gives $a_{n+1}/a_{n}\sim 4\zeta_{1}$ as $n\to\infty$.
One deduces that the condition given in (\ref{eq:Xp1}) is necessary and
sufficient for the absolute convergence of the series
$\sum a_{n}$, hence (\ref{eq:Xp1}).

The confluent Horn $\mrm{H}_{10}$ in (\ref{eq:Xp2}) is of the form
\[
\mrm{H}_{10}(a+n;b+n;\zeta_{1},\zeta_{3})=
\frac{\Gamma(b+n)}{\Gamma(a+n)}\sum_{m,p=0}^{\infty}
\frac{\zeta_{1}^{m}}{m!}\frac{\zeta_{3}^{p}}{p!}
\frac{\Gamma(a+2m-p+n)}{\Gamma(b+m+n)},
\]
$\abs{\zeta_{1}}<\frac{1}{4}$. Apply (\ref{eq:Gamma})
to $\Gamma(a+2m-p+n)$ and $\Gamma(b+m+n)$ to obtain
\begin{align*}
&\mrm{H}_{10}(a+n;b+n;\zeta_{1},\zeta_{3})=
\frac{\Gamma(b+n)}{\Gamma(a+n)}n^{a-b}
\sum_{m,p=0}^{\infty}
\frac{(n\zeta_{1})^{m}}{m!}\frac{(\zeta_{3}/n)^{p}}{p!} 
\\
&\times\exp\left(\sum_{k=1}^{\infty}
\frac{ (-1)^{k+1}n^{-k} }{ k(k+1) }
\left(B_{k+1}(a+2m-p)-B_{k+1}(b+m)\right)
\right).
\end{align*}
In the limit $n\to\infty$, $(\zeta_{3}/n)^{p}/p!\sim0$ for $p\in\N$,
and $=1$ for $p=0$. Hence, put $p=0$ in the above equation and
deduce that
\begin{equation}
\mrm{H}_{10}(a+n;b+n;\zeta_{1},\zeta_{3})\sim 
\,_{2}F_{1}\left(\begin{array}{c}
\frac{a+n}{2}, \frac{a+1+n}{2} \\ b+n
\end{array};4\zeta_{1}\right)
\quad\text{as}\quad n\to\infty.
\label{eq:H10sym2}
\end{equation}
Define, for convenience,
\begin{equation}
\alpha=\frac{a}{2},\quad
\beta=1+\frac{a}{2}-b,\quad 
\gamma=\frac{1}{2},\quad
\zeta=4\zeta_{1},\quad
\lambda=\frac{n}{2}.
\label{eq:H10sym3}
\end{equation}
In \cite[\S 7.2, Eq.~(11)]{Luke69} it was shown that,
for large complex $\lambda$,
\begin{align}
&\,_{2}F_{1}\left(\begin{array}{c}
\alpha+\lambda,1+\alpha-\gamma+\lambda \\ 
1+\alpha-\beta+2\lambda
\end{array};\zeta\right)\sim
\frac{ \zeta^{\beta-\lambda}\sqrt{\pi}
\Gamma(1+\alpha-\beta+2\lambda) }{ 
\sqrt{\lambda}
\Gamma(1+\alpha-\gamma+\lambda)
\Gamma(\gamma-\beta+\lambda) }
\nonumber \\
&\times
\frac{ \left( \zeta-1+\sqrt{1-\zeta}
\right)^{ \gamma-\alpha-\beta-\frac{1}{2} } }{ 
\left(1-\sqrt{1-\zeta} \right)^{ \gamma-\frac{1}{2} } }
\quad\text{as}\quad\abs{\lambda}\to\infty,
\label{eq:H10sym4}
\end{align}
$\aabs{\arg(\lambda)}\leq\pi-\delta$, $\delta>0$.
Substitute (\ref{eq:H10sym3})--(\ref{eq:H10sym4}) in
(\ref{eq:H10sym2}) and get that,
$\forall\abs{\zeta_{1}}<\frac{1}{4}$,
\begin{align*}
\mrm{H}_{10}(a+n;b+n;\zeta_{1},\zeta_{3})\sim &
2^{1-b+\frac{a}{2}}\zeta_{1}^{1-b}
\left(4\zeta_{1}-1+\sqrt{ 1-4\zeta_{1} } \right)^{ b-a-1 } \\
&\times\left(1-2\zeta_{1}-\sqrt{ 1-4\zeta_{1} }
\right)^{ \frac{a}{2} } 
\left(\frac{ \sqrt{ 1-2\zeta_{1}-\sqrt{ 1-4\zeta_{1} } } }{ 
\sqrt{2}\zeta_{1} }\right)^{n}
\end{align*}
as $n\to\infty$.
Define 
\[
b_{n}=\frac{ (\zeta_{1}\zeta_{2})^{n}
(a)_{n} }{ (b)_{n} }\mrm{H}_{10}(a+n;b+n;
\zeta_{1},\zeta_{3}),
\]
apply (\ref{eq:gamma}) to $(a)_{n}/(b)_{n}$,
and deduce from the above asymptotic formula that
the series $\sum b_{n}$ is absolutely convergent if
\[
\aaabs{\zeta_{2}\sqrt{ 1-2\zeta_{1}-\sqrt{ 1-4\zeta_{1} } }}
\leq\sqrt{2};
\] 
the equality holds only if
$\re(a-b)<0$, hence (\ref{eq:Xp2}). 
The condition is necessary and sufficient.

The complete Horn $H_{3}$ function in (\ref{eq:Xp3})
can be represented as follows
\begin{align*}
H_{3}(a-n,1;b;x,y)=&\sum_{m,p=0}^{\infty}
C_{mp}x^{m}y^{p},\quad 
C_{mp}=\frac{ \Gamma(a+2m+p-n)\Gamma(b) }{ 
m!\Gamma(a-n)\Gamma(b+m+p) }, \\
&x=\zeta_{1},\quad y=\zeta_{1}\zeta_{2},\quad
\abs{x}<R,\quad\abs{y}<S, \\
&R+\left(S-\frac{1}{2}\right)^{2}=\frac{1}{4}.
\end{align*}
It will be shown that the condition for the convergence of
$H_{3}$ is necessary and sufficient
for the absolute convergence of the third series 
(\ref{eq:Xp3}).

Extract $H_{3}$ as a sum $X_{n}+Y_{n}$, where
\[
X_{n}=\sum_{m=0}^{\infty}\sum_{p=0}^{n}
C_{mp}x^{m}y^{p},\quad 
Y_{n}=\sum_{m=0}^{\infty}\sum_{p=n+1}^{\infty}
C_{mp}x^{m}y^{p}.
\]
The first term explicitly reads
\begin{align*}
X_{n}=&\frac{\Gamma(b)}{\Gamma(a-n)}
\sum_{m=0}^{\infty}\frac{x^{m}}{m!}\left(
\frac{\Gamma(a+2m-n)}{\Gamma(b+m)}+
\frac{\Gamma(a+2m+1-n)}{\Gamma(b+m+1)}y\right. \\
&+\frac{\Gamma(a+2m+2-n)}{\Gamma(b+m+2)}y^{2}
+\ldots+\frac{\Gamma(a+2m-1)}{\Gamma(b+m+n-1)}
y^{n-1} \\
&\left.+\frac{\Gamma(a+2m)}{\Gamma(b+m+n)}
y^{n}\right) \\
\sim&\frac{\Gamma(b)}{\Gamma(a-n)}
\sum_{m=0}^{\infty}\frac{x^{m}}{m!}
\frac{\Gamma(a+2m-n)}{\Gamma(b+m)}
\quad\text{as}\quad n\to\infty
\end{align*}
for it holds $\abs{y}<1$. Hence,
\begin{equation}
X_{n}\sim\, _{2}F_{1}\left(\begin{array}{c}
\frac{a-n}{2},\frac{a+1-n}{2} \\
b
\end{array};4x\right)
\quad\text{as}\quad n\to\infty,
\quad\abs{x}<\frac{1}{4}.
\label{eq:Xn}
\end{equation}
By the Kummer transformation formula
\cite[\S 3.8, Eqs.~(1)--(3)]{Luke69},
\begin{equation}
_{2}F_{1}\left(\begin{array}{c}
\frac{a-n}{2},\frac{a+1-n}{2} \\
b
\end{array};4x\right)=
(1-4x)^{\frac{ n-a-1 }{2}}\,
_{2}F_{1}\left(\begin{array}{c}
b-\frac{a-n}{2},\frac{a+1-n}{2} \\
b
\end{array};\frac{4x}{4x-1}\right).
\label{eq:Xn-1}
\end{equation}
Define, for convenience,
\begin{equation}
\alpha=b-\frac{a}{2},\quad \beta=\frac{a+1}{2},\quad
\gamma=b,\quad \zeta=\frac{1+4x}{1-4x},\quad
\lambda=\frac{n}{2}.
\label{eq:Xn-2}
\end{equation}
Then, for large complex $\lambda$, it was shown that
\cite[\S 7.2, Eq.~(8)]{Luke69},
\begin{align}
&_{2}F_{1}\left(\begin{array}{c}
\alpha+\lambda,\beta-\lambda \\
\gamma
\end{array};\frac{1-\zeta}{2}\right)\sim 
\frac{ 2^{\alpha+\beta-1}\Gamma(1-\beta+\lambda)
\Gamma(\gamma) }{ 
\sqrt{\pi\lambda}\Gamma(\gamma-\beta+\lambda) }
\nonumber \\
&\times\frac{ \left(1+\zeta-\sqrt{\zeta^{2}-1} \right)^{
\gamma-\alpha-\beta-\frac{1}{2} } }{ 
\left(1-\zeta+\sqrt{\zeta^{2}-1} \right)^{
\gamma-\frac{1}{2} } }\left(
\left(\zeta+\sqrt{\zeta^{2}-1}\right)^{\lambda-\beta}\right.
\nonumber \\
&\left.+\e^{\pm\im\pi\left(\gamma-\frac{1}{2}\right) }
\left(\zeta+\sqrt{\zeta^{2}-1}\right)^{-\lambda-\alpha}\right)
\quad\text{as}\quad\abs{\lambda}\to\infty
\label{eq:Xn-3}
\end{align}
where the upper (lower) sign in the exponent is taken if
$\img\zeta>0$ ($\img\zeta<0$), and
\begin{align*}
&-\frac{\pi}{2}-\omega_{2}+\delta<
\arg(\lambda)<\frac{\pi}{2}+\omega_{1}-\delta,
\quad\delta>0, \\
&\omega_{1}=-\arctan\left(\frac{\nu-\pi}{\mu}\right),\quad
\omega_{2}=\arctan\left(\frac{\nu}{\mu}\right),\quad
\nu\geq0, \\
&\omega_{1}=-\arctan\left(\frac{\nu}{\mu}\right),\quad
\omega_{2}=\arctan\left(\frac{\nu+\pi}{\mu}\right),\quad
\nu\leq0,
\end{align*}
and $\mu=\re\xi$, $\nu=\img\xi$, $\zeta=\cosh\xi$.

In the second term, $Y_{n}$, make a substitution
$p=n+1+q$, $q=0,1,\ldots$ and apply
\[
\Gamma(a-n)=\frac{ (-1)^{n}\Gamma(a)
\Gamma(1-a) }{ \Gamma(1-a+n) }.
\]
Then
\[
Y_{n}=\frac{ (-1)^{n}y^{n+1}\Gamma(b) }{ 
\Gamma(a)\Gamma(1-a) }
\sum_{m,q=0}^{\infty}\frac{x^{m}}{m!}y^{q}
\frac{ \Gamma(1-a+n)\Gamma(a+1+2m+q) }{ 
\Gamma(b+m+1+q+n) }.
\]
Apply (\ref{eq:gamma}) and get that
\begin{align*}
Y_{n}\sim& \frac{ (-1)^{n}y^{n+1}\Gamma(b) }{ 
\Gamma(a)\Gamma(1-a) }
\sum_{m,q=0}^{\infty}\frac{x^{m}}{m!}y^{q}
n^{-a-b-m-q}\Gamma(a+1+2m+q)
\end{align*}
as $n\to\infty$. Hence,
\begin{equation}
Y_{n}\sim \frac{a\Gamma(b)}{\Gamma(1-a)}
(-1)^{n}y^{n+1}n^{-a-b}\quad
\text{as}\quad n\to\infty.
\label{eq:Yn}
\end{equation}
Substitute (\ref{eq:Xn-1})--(\ref{eq:Xn-3})
in (\ref{eq:Xn}) and exploit
(\ref{eq:gamma}) to get the asymptotic for $X_{n}$,
then substitute this formula in $H_{3}=X_{n}+Y_{n}$,
where $Y_{n}$ is as in (\ref{eq:Yn}), and get that
\begin{align*}
H_{3}(a-n,1;b;x,y)\sim&
\frac{ 2^{2b-1}\Gamma(b) }{ \sqrt{\pi} }
\frac{ (1-4x)^{ \frac{n-a-1}{2} }
\left(\frac{1}{2}+\sqrt{x} \right)^{ b+\frac{1}{2} } }{ 
\left(2\sqrt{x} \right)^{ b-\frac{1}{2} } } \\
&\times n^{ \frac{1}{2}-b }\left(
\eta(x)^{ \frac{n-a-1}{2} }+\varphi(x)
\eta(x)^{ \frac{a-n}{2}-b }\right) \\
&+\frac{a\Gamma(b)}{\Gamma(1-a)}
(-1)^{n}y^{n+1}n^{-a-b}\quad
\text{as}\quad n\to\infty, \\
\eta(x)=&\frac{ \frac{1}{2}+\sqrt{x} }{ \frac{1}{2}-\sqrt{x} },
\quad
\varphi(x)=\exp\left(\im\pi\left(b-\frac{1}{2}\right)
\mrm{sgn}\img\left(\frac{1+4x}{1-4x} \right)\right).
\end{align*}
Define
\[
c_{n}=\frac{ (-\zeta_{3})^{n} }{ n!(1-a)_{n} }
H_{3}(a-n,1;b;x,y)
\]
and thus deduce that $c_{n}\sim0$ as $n\to\infty$
for all $x$, $y$ in the cone
$R+\left(S-\frac{1}{2}\right)^{2}=\frac{1}{4}$.
This completes the proof of (\ref{eq:Xp123}).

The proof of (\ref{eq:recurrXp})--(\ref{eq:conflXp})
is straightforward.\qedhere
\end{step}
\end{proof}
The proof of Theorem~\ref{thm:Integral} relies on the following
result.
\begin{lem}\label{lem:K1K2}
\begin{align*}
\sum_{m=0}^{\infty}\sum_{n=0}^{m}&
\frac{ m!x^{m}y^{n} }{ (2m+1)!(m-n)! }
K_{2m-n+\frac{1}{2}}(z) \\
=&\sqrt{\frac{\pi}{2z}}X^{\prime}\left(
\frac{1}{2},\frac{3}{2};\frac{x}{z^{2}},\frac{y z}{2},
-\frac{z^{2}}{4}\right)
-\sqrt{\frac{\pi z}{2}}X\left(
\frac{3}{2},\frac{3}{2};\frac{x z^{2}}{16},-\frac{x y z}{8},
\frac{z^{2}}{4}\right), \\
\sum_{m=0}^{\infty}\sum_{n=0}^{m}&
\frac{ m!x^{m}y^{n} }{ (2m)!(m-n)! }
K_{2m-n-\frac{1}{2}}(z) \\
=&\sqrt{\frac{\pi}{2z}}X\left(
\frac{1}{2},\frac{1}{2};\frac{x z^{2}}{16},-\frac{x y z}{8},
\frac{z^{2}}{4}\right)
-\sqrt{\frac{\pi z}{2}}X^{\prime}\left(
-\frac{1}{2},\frac{1}{2};\frac{x}{z^{2}},\frac{y z}{2},
-\frac{z^{2}}{4}\right).
\end{align*}
The variables $(x,y,z)\in\C^{3}$, with $\abs{z}<\infty$, 
fulfill at least one of the following three conditions:
\begin{enumerate}[\upshape ($i$)]
\item $4\abs{x}\leq\abs{z}^{2}$, $\abs{yz}<4$
\item $4\abs{x}<\abs{z}^{2}$, 
$\abs{x y}\leq \abs{z}+\sqrt{ \abs{z}^{2}-4\abs{x} }$
\item $\abs{x}<R\abs{z}^{2}$,
$\abs{xy}<2S\abs{z}$, $R+\left(S-\frac{1}{2}\right)^{2}=\frac{1}{4}$.
\end{enumerate}
Items ($i$), ($ii$) and ($iii$) indicate that 
$X^{\prime}$ obeys the series representation given in
(\ref{eq:Xp1}), (\ref{eq:Xp2}) and (\ref{eq:Xp3}), respectively.
The series representation for $X$ admits any form
given in (\ref{eq:X123}).
\end{lem}
Here $K_{2m-n\pm\frac{1}{2}}$ is the Macdonald function.
\begin{proof}
First, note that ($i$), ($ii$) and ($iii$) are due to
(\ref{eq:Xp1}), (\ref{eq:Xp2}) and (\ref{eq:Xp3}), respectively, by
setting $\zeta_{1}=\frac{x}{z^{2}}$, $\zeta_{2}=\frac{yz}{2}$,
$\zeta_{3}=-\frac{z^{2}}{4}$. The conditions ensure that both
series in the lemma are absolutely convergent.

Substitute the series representation \cite[\S V.5.3, Eq.~(3)]{Vilenkin91}
of Macdonald function
\[
K_{\nu}(z)=
\frac{1}{2}\sum_{l=0}^{\infty}
\frac{ (-1)^{l}\Gamma(\nu-l) }{ l! }\left(\frac{z}{2}\right)^{-\nu+2l}+
\frac{1}{2}\sum_{l=0}^{\infty}
\frac{ (-1)^{l}\Gamma(-\nu-l) }{ l! }\left(\frac{z}{2}\right)^{\nu+2l}
\]
($\nu=2m-n+\frac{1}{2}$) in the left-hand side of the
first series in the lemma and get the expression
\begin{align*}
&\frac{1}{\sqrt{ 2z }}\sum_{m,n,l=0}^{\infty}
\frac{ (4x/z^{2})^{m} }{m!}\frac{ (yz/2)^{n} }{n!}
\frac{ (-z^{2}/4)^{l} }{l!}\left(
\frac{ (m!)^{2}n!\Gamma\left(2m-n+\frac{1}{2}-l \right) }{ 
(2m+1)!\Gamma(m-n+1) }\right) \\
&+\frac{1}{2}\sqrt{\frac{z}{2}}\sum_{m,n,l=0}^{\infty}
\frac{ (xz^{2}/4)^{m} }{m!}\frac{ (2y/z)^{n} }{n!}
\frac{ (-z^{2}/4)^{l} }{l!}\left(
\frac{ (m!)^{2}n!\Gamma\left(-2m+n-\frac{1}{2}-l \right) }{ 
(2m+1)!\Gamma(m-n+1) }\right).
\end{align*}
Write
\begin{align*}
\frac{ (m!)^{2}n!\Gamma\left(2m-n+\frac{1}{2}-l \right) }{ 
(2m+1)!\Gamma(m-n+1) }=&
\frac{ (1)_{m}(1)_{n}\left(\frac{1}{2}\right)_{2m-n-l}\sqrt{\pi} }{ 
4^{m}\left(\frac{3}{2}\right)_{m}(1)_{m-n} }, \\
\frac{ (m!)^{2}n!\Gamma\left(-2m+n-\frac{1}{2}-l \right) }{ 
(2m+1)!\Gamma(m-n+1) }=&
\frac{ 2\sqrt{\pi}(-1)^{n+l+1}(1)_{m}(1)_{n} }{ 4^{m}(1)_{m-n}
\left(\frac{3}{2}\right)_{m}\left(\frac{3}{2}\right)_{2m-n+l} }
\end{align*}
and get the expression
\begin{align*}
&\sqrt{\frac{\pi}{2z}}\sum_{m,n,l=0}^{\infty}
\frac{ (x/z^{2})^{m} }{ m! }\frac{ (yz/2)^{n} }{ n! }
\frac{ (-z^{2}/4)^{l} }{ l! }
\frac{ (1)_{m}(1)_{n}\left(\frac{1}{2}\right)_{2m-n-l} }{ 
\left(\frac{3}{2}\right)_{m}(1)_{m-n} } \\
&-\sqrt{\frac{\pi z}{2}}\sum_{m,n,l=0}^{\infty}
\frac{ (xz^{2}/16)^{m} }{ m! }\frac{ (-2y/z)^{n} }{ n! }
\frac{ (z^{2}/4)^{l} }{ l! }
\frac{ (1)_{m}(1)_{n} }{ 
(1)_{m-n}\left(\frac{3}{2}\right)_{m}\left(\frac{3}{2}\right)_{2m-n+l} }.
\end{align*}
The first triple series represents $X^{\prime}$ (\ref{eq:Xp}) with
\[
a=\frac{1}{2},\quad
b=\frac{3}{2},\quad
\zeta_{1}=\frac{x}{z^{2}},\quad
\zeta_{2}=\frac{yz}{2},\quad
\zeta_{3}=-\frac{z^{2}}{4}.
\]
The second triple series can be rewritten thus:
Make a substitution $l\to l-m$ for $l=m,m+1,\ldots$
and get the series
\[
\sum_{m,n=0}^{\infty}
\sum_{l=m}^{\infty}
\frac{ (x/4)^{m} }{ m! }\frac{ (-2y/z)^{n} }{ n! }
\frac{ (z^{2}/4)^{l} }{ l! } 
\frac{ (1)_{m}(1)_{n}(1)_{l} }{ 
(1)_{m-n}\left(\frac{3}{2}\right)_{m}\left(\frac{3}{2}\right)_{m-n+l}
(1)_{l-m} }.
\]
But $1/(1)_{l-m}=0$ for $l=0,1,\ldots,m-1$. Thus, the sum over
$l=m,m+1,\ldots$ can be replaced with the sum over
$l=0,1,\ldots$. Next, let $m=p+n$, $l=m+q$. Then $l=n+p+q$,
$p=-n,-n+1,\ldots$, $q=0,1,\ldots$ and the general term of 
the above series obeys the form
\[
\frac{a^{m}}{m!}\frac{b^{n}}{n!}\frac{c^{l}}{l!}
\frac{ (1)_{m}(1)_{n}(1)_{l} }{ 
(1)_{m-n}\left(\frac{3}{2}\right)_{m}\left(\frac{3}{2}\right)_{m-n+l}
(1)_{l-m} } 
=\frac{(ac)^{p}}{p!}\frac{(abc)^{n}}{n!}\frac{c^{q}}{q!}
\frac{ (1)_{n} }{ 
\left(\frac{3}{2}\right)_{n+p}\left(\frac{3}{2}\right)_{2p+n+q} }
\]
($a=x/4$, $b=-2y/z$, $c=z^{2}/4$). But
$1/p!=0$ for $p=-n,-n+1,\ldots,-1$. Hence, the sum over
$p=-n,-n+1,\ldots$ can be replaced with the sum over 
$p=0,1,\ldots$ Relabeling $p$ with $m$ and $q$ with $l$, one derives the series
\[
\sum_{m,n,l=0}^{\infty}
\frac{ (xz^{2}/16)^{m} }{ m! }\frac{ (-xyz/8)^{n} }{ n! }
\frac{ (z^{2}/4)^{l} }{ l! }
\frac{ (1)_{n} }{ \left(\frac{3}{2}\right)_{2m+n+l}
\left(\frac{3}{2}\right)_{m+n} }
\]
which is $X$ (\ref{eq:X}) with 
\[
a=\frac{3}{2},\quad
b=\frac{3}{2},\quad
\zeta_{1}=\frac{x z^{2}}{16},\quad
\zeta_{2}=-\frac{x y z}{8},\quad 
\zeta_{3}=\frac{z^{2}}{4}.
\]
This demonstrates the first series in the lemma.
The proof
of the second one is omitted. For, the derivation of the second
series requires no additional ideas but those presented above
(in the series representation of Macdonald function put
$\nu=2m-n-\frac{1}{2}$ and proceed identically as before). 
\end{proof}
To accomplish the proof of the theorem, one needs to show that
the integrals (\ref{eq:Integral1})--(\ref{eq:Integral2}) 
obey the series representation given in
Lemma~\ref{lem:K1K2}.

For this, rewrite $Q(p)^{-1}$ with the help of
binomial series
\begin{align*}
\frac{1}{\left(p^{2}-\zeta\right)^{2}-
\alpha^{2}\left(p_{1}^{2}+p_{2}^{2}\right)-\beta^{2}}=&
\sum_{m=0}^{\infty}
\frac{ \left(\beta^{2}+\alpha^{2}
\left(p_{1}^{2}+p_{2}^{2}\right)\right)^{m} }{ 
\left(p^{2}-\zeta\right)^{2m+2} }, \\
\left(\beta^{2}+\alpha^{2}
\left(p_{1}^{2}+p_{2}^{2}\right)\right)^{m}=&
\sum_{n=0}^{m}\binom{m}{n}\alpha^{2n}\beta^{2m-2n}
\left(p_{1}^{2}+p_{2}^{2}\right)^{n},\quad \forall m\in\N_{0}.
\end{align*}
The convergence of the series is ensured by the convergence of $X^{\prime}$
as it will be shown below.

Rewrite $p=(p_{1},p_{2},p_{3})\in\R^{3}$ in spherical coordinates
$(k,\vartheta,\varphi)$, $k=\abs{p}$, and
substitute the series representation of $Q(p)^{-1}$ in the 
left-hand side of (\ref{eq:Integral1}). Then
\begin{align*}
\text{(\ref{eq:Integral1})}=&
\frac{1}{(2\pi)^{3}}\int_{0}^{2\pi}\dd\varphi\int_{0}^{\pi}\dd\vartheta\,
\sin\vartheta\,\int_{0}^{\infty}\dd k\,k^{2}\e^{-\im kr\cos\vartheta} \\
&\times \sum_{m=0}^{\infty}
\frac{1}{\left(k^{2}-\zeta\right)^{2m+2}}\sum_{n=0}^{m}
\binom{m}{n}\alpha^{2n}\beta^{2m-2n}
k^{2n}\sin^{2n}\vartheta.
\end{align*}
But
\begin{align*}
\int_{0}^{\pi}\dd\vartheta\,\e^{-\im kr\cos\vartheta}
\sin^{2n+1}\vartheta=&
\frac{\sqrt{\pi}n!}{\Gamma\left(n+\frac{3}{2}\right)}
\,_{0}F_{1}\left(;n+\frac{3}{2};-\frac{1}{4}k^{2}r^{2}\right), \\
\int_{0}^{\infty}\dd k\,k^{2n+2}\frac{ _{0}F_{1}\left(;
n+\frac{3}{2};-\frac{1}{4}k^{2}r^{2}\right) }{ 
\left(k^{2}-\zeta\right)^{2m+2} }=&
\frac{\Gamma\left(n+\frac{3}{2}\right)}{(2m+1)!}
(-\zeta)^{ -\frac{1}{4}-m+n/2 }
\left(\frac{r}{2}\right)^{\frac{1}{2}+2m-n} \\
&\times K_{2m-n+\frac{1}{2}}\left(r\sqrt{-\zeta}\right).
\end{align*}
Hence,
\[
\text{(\ref{eq:Integral1})}=
\frac{\sqrt{r}}{4\pi \sqrt{2\pi}\sqrt[4]{-\zeta} }
\sum_{m=0}^{\infty}\sum_{n=0}^{m}
\frac{ m!x^{m}y^{n} }{ (2m+1)!(m-n)! }
K_{2m-n+\frac{1}{2}}(z)
\]
with 
\begin{equation}
x=-\frac{\beta^{2}r^{2}}{4\zeta},\quad
y=\frac{ 2\alpha^{2}\sqrt{-\zeta} }{ \beta^{2}r },\quad
z=r\sqrt{-\zeta}.
\label{eq:xyz}
\end{equation}
Apply the first series in Lemma~\ref{lem:K1K2} and get 
the expression as in the right-hand side of (\ref{eq:Integral1}).

The calculation of the second integral, (\ref{eq:Integral2}), 
is similar, but in this case one infers the integral
\begin{align*}
\int_{0}^{\infty}\dd k\,k^{2n+2}\frac{ _{0}F_{1}\left(;
n+\frac{3}{2};-\frac{1}{4}k^{2}r^{2}\right) }{ 
\left(k^{2}-\zeta\right)^{2m+1} }=&
\frac{\Gamma\left(n+\frac{3}{2}\right)}{(2m)!}
(-\zeta)^{ \frac{1}{4}-m+n/2 }
\left(\frac{r}{2}\right)^{-\frac{1}{2}+2m-n} \\
&\times K_{2m-n-\frac{1}{2}}\left(r\sqrt{-\zeta}\right)
\end{align*}
due to the numerator $p^{2}-\zeta$. Hence,
\[
\text{(\ref{eq:Integral2})}=
\frac{\sqrt[4]{-\zeta}}{2\pi \sqrt{2\pi r} }
\sum_{m=0}^{\infty}\sum_{n=0}^{m}
\frac{ m!x^{m}y^{n} }{ (2m)!(m-n)! }
K_{2m-n-\frac{1}{2}}(z)
\]
with $x$, $y$, $z$ as in (\ref{eq:xyz}).
Apply the second series in Lemma~\ref{lem:K1K2} and get 
the expression as in right-hand side of (\ref{eq:Integral2}). 

Next, substitute
\begin{equation}
\zeta_{1}=\frac{\beta^{2}}{4\zeta^{2}},\quad
\zeta_{2}=-\frac{\zeta\alpha^{2}}{\beta^{2}},\quad 
\zeta_{3}=\frac{\zeta r^{2}}{4}
\label{eq:zeta123}
\end{equation}
in the conditions in (\ref{eq:Xp123}) and get that
(\ref{eq:Xp1})$\Rightarrow$($a$),
(\ref{eq:Xp2})$\Rightarrow$($b$),
(\ref{eq:Xp3})$\Rightarrow$($c$).
This completes
the proof of the theorem and the main results as a whole.
\section{Further properties and corollaries}\label{sec:6}
In this paragraph, some further properties of the series
defined in Lemmas~\ref{lem:X}--\ref{lem:Xp} are discussed.
In particular, the results will be useful for the derivation of
the series representation for the off-diagonal terms of
Green's function (\ref{eq:Green}). These terms will be discussed
in the subsequent paragraph \S \ref{sec:7}.
\begin{prop}\label{prop:X}
Consider $\zeta=(\zeta_{1},\zeta_{2},\zeta_{3})\in\C^{3}$
and define the differential operators
$\partial_{j}=\partial/\partial\zeta_{j}$,
$\partial_{jk}=\partial_{j}\partial_{k}$,
$\forall j,k=1,2,3$. Then,
$\forall a,b\in\C-\{-n\co n\in\N_{0}\}$,
\begin{align}
&\partial_{1}X(a,b;\zeta)=
\frac{1}{ ab(a+1) }X(a+2,b+1;\zeta), \label{eq:pd1X} \\
&\partial_{2}X(a,b;\zeta)=
\frac{1}{ab}X(a+1,b+1;\zeta)
+\frac{\zeta_{2}}{ab}\partial_{2}X(a+1,b+1;\zeta), \label{eq:pd2X} \\
&\partial_{3}X(a,b;\zeta)=
\frac{1}{a}X(a+1,b;\zeta), \label{eq:pd3X} \\
&\left(\partial_{1}+\zeta_{2}\partial_{12}-\partial_{23}\right)
X(a,b;\zeta)=0. \label{eq:diffX}
\end{align}
\end{prop}
\begin{proof}
To prove equations~(\ref{eq:pd1X})--(\ref{eq:pd3X}),
explore the definition~(\ref{eq:X}) and elementary properties of Pochhammer symbol;
(\ref{eq:pd1X})--(\ref{eq:pd3X})$\Rightarrow$(\ref{eq:diffX}).
\end{proof}
\begin{cor}\label{cor:X}
Let $\zeta=(\zeta_{1},\zeta_{2},\zeta_{3})\in\C^{3}$.
Then, $\forall a,b\in\C-\{-n\co n\in\N_{0}\}$,
\begin{align}
0=&\sum_{n=1}^{\infty}
\frac{\zeta_{2}^{n}}{ (a)_{n}(b)_{n} }
\Biggl(X(a+n,b+n;\zeta) \nonumber \\
&-n\,F_{1:1;0}^{0:0;0}\left(
\begin{array}{rr}
: & ;; \\
\left(a+n:2,1\right): & \left(b+n:1\right);;
\end{array}\zeta_{1},\zeta_{3}
\right)\Biggr), \label{eq:Xgener}
\end{align}
provided each series involved is absolutely convergent. 
\end{cor}
\begin{proof}
For convenience, define
\[
F(a,b;\zeta_{1},\zeta_{3})=
\,F_{1:1;0}^{0:0;0}\left(
\begin{array}{rr}
: & ;; \\
\left(a:2,1\right): & \left(b:1\right);;
\end{array}\zeta_{1},\zeta_{3}
\right).
\]
By (\ref{eq:pd2X}),
\begin{align*}
\partial_{2}X(a,b;\zeta)=&
\frac{1}{ab}X(a+1,b+1;\zeta)+
\frac{\zeta_{2}}{a(a+1)b(b+1)}X(a+2,b+2;\zeta)+
\ldots \\
=&\frac{1}{ab}
\sum_{n=0}^{\infty}\frac{\zeta_{2}^{n}}{ (a+1)_{n}
(b+1)_{n} }X(a+1+n,b+1+n;\zeta).
\end{align*}
Then, by (\ref{eq:recurrX}),
\begin{align*}
\partial_{2}X(a,b;\zeta)=&
\frac{1}{ab}
\sum_{n=0}^{\infty}\frac{\zeta_{2}^{n}}{ (a+1)_{n}
(b+1)_{n} }\frac{ (a+n)(b+n) }{ \zeta_{2} } \\
&\times\left(X(a+n,b+n;\zeta)-F(a+n,b+n;\zeta_{1},\zeta_{3})\right) \\
=&\frac{1}{\zeta_{2}}
\sum_{n=0}^{\infty}\frac{\zeta_{2}^{n}}{ (a)_{n}(b)_{n} }
X(a+n,b+n;\zeta)-\frac{1}{\zeta_{2}}X(a,b;\zeta) \\
\Rightarrow&\partial_{2}X(a,b;\zeta)=
\sum_{n=1}^{\infty}\frac{\zeta_{2}^{n-1}}{ (a)_{n}(b)_{n} }
X(a+n,b+n;\zeta).
\end{align*}
But also
\begin{equation}
\partial_{2}X(a,b;\zeta)=
\sum_{n=1}^{\infty}\frac{n\zeta_{2}^{n-1}}{ (a)_{n}(b)_{n} }
F(a+n,b+n;\zeta_{1},\zeta_{3})
\label{eq:pd2X-0}
\end{equation}
by (\ref{eq:X2}). Hence, one derives (\ref{eq:Xgener}).
\end{proof}
\begin{exam}
Substitute $X$ (\ref{eq:X1}) in (\ref{eq:Xgener}) and get that
\begin{align}
\sum_{m,n=0}^{\infty}&
\frac{ \zeta_{1}^{m}\zeta_{2}^{n} }{ 
m!(a)_{2m+n}(b)_{m+n} }
F_{1:1;0}^{0:1;0}\left(
\begin{array}{rr}
: & 1;; \\
a+n+2m: & b+n+m;;
\end{array}\zeta_{2},\zeta_{3}\right) \nonumber \\
=&\sum_{n=0}^{\infty}
\frac{ (n+1)\zeta_{2}^{n} }{ (a)_{n}(b)_{n} }
F_{1:1;0}^{0:0;0}\left(
\begin{array}{rr}
: & ;; \\
\left(a+n:2,1\right): & \left(b+n:1\right);;
\end{array}\zeta_{1},\zeta_{3}\right). \label{eq:Xgener-1}
\end{align}
In particular, set $\zeta_{1}=\zeta_{3}=0$, $\zeta_{2}=\zeta$
in (\ref{eq:Xgener-1}) and get a well-known series identity
\[
\sum_{n=1}^{\infty}\frac{ \zeta^{n} }{ (a)_{n}(b)_{n} }\,
_{1}F_{2}\left(\begin{array}{c}
1 \\
a+n,b+n
\end{array};\zeta\right)=\frac{\zeta}{ab}\,
_{1}F_{2}\left(\begin{array}{c}
2 \\
a+1,b+1
\end{array};\zeta\right).
\]
[The above equation is
easy to prove by representing $_{1}F_{2}$ on the left-hand side as a
series and then by applying the sum rule
\cite[\S 1.6, Eq.~(20)]{Srivastava84}.]

Set $\zeta_{2}=0$ in (\ref{eq:Xgener-1}) and get, after relabeling
the parameters (that is, $\zeta_{1}\to \frac{1}{4}tz^{2}$, $\sqrt{\zeta_{3}}
\to \frac{1}{2}z$, $a\to \nu+1$, $b\to a$), that
\[
\sum_{n=0}^{\infty}\frac{ t^{n}I_{2n+\nu}(z) }{ n!(a)_{n} }=
\frac{(z/2)^{\nu}}{\Gamma(\nu+1)}\,
F_{1:1;0}^{0:0;0}\left(
\begin{array}{rr}
: & ;; \\
\left(\nu+1:2,1\right): & \left(a:1\right);;
\end{array}\frac{t z^{2}}{4},\frac{z^{2}}{4}\right)
\]
where $I_{2n+\nu}$ is the modified Bessel function.
Many more series identities can be deduced from (\ref{eq:Xgener-1})
if one notes that, for $\zeta_{2}=-\zeta_{3}$, the Kamp\'{e}
de F\'{e}riet function $F_{1:1;0}^{0:1;0}$ reduces to
the hypergeometric function $_{1}F_{2}$
\cite[\S 3.2]{Kim14}.
\end{exam}
\begin{prop}\label{prop:Xp}
Let $\zeta=(\zeta_{1},\zeta_{2},\zeta_{3})\in\C^{3}$
and define the differential operators 
$\partial_{j}=\partial/\partial\zeta_{j}$,
$D_{j}=\zeta_{j}\partial_{j}$, $\forall j=1,2,3$.
Then,
$\forall a\in\C-\N$, $\forall b\in\C-\{-n\co n\in\N_{0}\}$,
\begin{align}
\partial_{1}X^{\prime}(a,b;\zeta)=&
\frac{a}{b}\zeta_{2}\left(1+D_{1}\right)
X^{\prime}(a+1,b+1;\zeta) \nonumber \\
&+\frac{a(a+1)}{b}\mrm{H}_{10}(a+2;b+1;\zeta_{1},\zeta_{3}),
\label{eq:pd1Xp} \\
\partial_{2}X^{\prime}(a,b;\zeta)=&
\frac{a}{b}\zeta_{1}\left(1+D_{2}\right)
X^{\prime}(a+1,b+1;\zeta), \label{eq:pd2Xp} \\
\partial_{3}X^{\prime}(a,b;\zeta)=&
\frac{1}{a-1}X^{\prime}(a-1,b;\zeta), \label{eq:pd3Xp} \\
0=&\left(\left(D_{3}+D_{2}-2D_{1}+1-a\right)\partial_{3}+1
\right)X^{\prime}(a,b;\zeta). \label{eq:diffXp}
\end{align}
\end{prop}
\begin{proof}
To prove equations~(\ref{eq:pd1Xp})--(\ref{eq:pd3Xp}),
explore the definition~(\ref{eq:Xp}) and elementary properties of Pochhammer symbol.
The implication (\ref{eq:pd1Xp})--(\ref{eq:pd3Xp})$\Rightarrow$(\ref{eq:diffXp})
is not obvious, and thus the proof of (\ref{eq:diffXp}) will be outlined below.

The easiest way to show (\ref{eq:diffXp}) is to begin with the
equation 
\begin{align*}
\left(\left(D_{3}+1-a\right)\partial_{3}+1\right)&
\mrm{H}_{10}(a+n;b+n;\zeta_{1},\zeta_{3}) \\
&=(n+2D_{1})\partial_{3}\mrm{H}_{10}(a+n;b+n;\zeta_{1},\zeta_{3})
\end{align*}
which proceeds from \cite[\S 5.9, Eq.~(41)]{Erdelyi53},
\cite[Appendix~A.2, Eq.~(A.19$^{\prime}$)]{Debiard03}. Then,
apply the operator $\left(D_{3}+1-a\right)\partial_{3}+1$
to (\ref{eq:Xp2}) to obtain the expression
\begin{align*}
\left(\left(D_{3}+1-a\right)\partial_{3}+1\right)&
X^{\prime}(a,b;\zeta)=\partial_{3}
\sum_{n=1}^{\infty}\frac{ n(\zeta_{1}\zeta_{2})^{n}
(a)_{n} }{ (b)_{n} }\mrm{H}_{10}(a+n;b+n;\zeta_{1},\zeta_{3}) \\
&+2\zeta_{1}\partial_{3}
\sum_{n=0}^{\infty}\frac{ (\zeta_{1}\zeta_{2})^{n}
(a)_{n} }{ (b)_{n} }\partial_{1}
\mrm{H}_{10}(a+n;b+n;\zeta_{1},\zeta_{3}).
\end{align*}
The first sum on the right-hand side is just $D_{2}
X^{\prime}$, as it is seen from (\ref{eq:Xp2}). In the second
sum, substitute 
\[
\zeta_{1}^{n}\partial_{1}\mrm{H}_{10}=
\partial_{1}\left(\zeta_{1}^{n}\mrm{H}_{10}\right)-
n\zeta_{1}^{n-1}\mrm{H}_{10}
\]
and get that the sum ends up as 
$\partial_{1}X^{\prime}-\frac{1}{\zeta_{1}}
D_{2}X^{\prime}$. Hence,
\begin{align*}
\left(\left(D_{3}+1-a\right)\partial_{3}+1\right)
X^{\prime}(a,b;\zeta)=&\partial_{3}D_{2}
X^{\prime}(a,b;\zeta) \\
&+2\zeta_{1}\partial_{3}\left(
\partial_{1}X^{\prime}(a,b;\zeta)-\frac{1}{\zeta_{1}}
D_{2}X^{\prime}(a,b;\zeta)\right).
\end{align*}
This proves (\ref{eq:diffXp}).
\end{proof}
\begin{cor}\label{cor:Xp}
Let $\zeta=(\zeta_{1},\zeta_{2},\zeta_{3})\in\C^{3}$.
Then, $\forall a\in\C-\N$, $\forall b\in\C-\{-n\co n\in\N_{0}\}$,
\begin{equation}
0=
\sum_{n=1}^{\infty}\frac{ (\zeta_{1}\zeta_{2})^{n}
(a)_{n} }{ (b)_{n} }\left(X^{\prime}(a+n,b+n;\zeta)-
n\mrm{H}_{10}(a+n;b+n;\zeta_{1},\zeta_{3})\right),
\label{eq:Xpgener}
\end{equation}
provided each series involved is absolutely convergent.
\end{cor}
\begin{proof}
Similar to the proof of Corollary~\ref{cor:X},
exploit (\ref{eq:pd2Xp}) to obtain
\[
\partial_{2}X^{\prime}(a,b;\zeta)=
\frac{1}{\zeta_{2}}
\sum_{n=1}^{\infty}\frac{ (\zeta_{1}\zeta_{2})^{n}
(a)_{n} }{ (b)_{n} }X^{\prime}(a+n,b+n;\zeta).
\]
But
\[
\partial_{2}X^{\prime}(a,b;\zeta)=
\frac{1}{\zeta_{2}}
\sum_{n=1}^{\infty}\frac{ n(\zeta_{1}\zeta_{2})^{n}
(a)_{n} }{ (b)_{n} }\mrm{H}_{10}(a+n;b+n;\zeta_{1},\zeta_{3})
\]
by (\ref{eq:Xp2}); hence the result.
\end{proof}
\section{Series representation for Green's function}\label{sec:7}
The diagonal terms $G_{2}\pm\beta G_{1}$
of Green's function $\mathcal{G}_{R}$
(\ref{eq:Green}) obey the series representation due to
Theorem~\ref{thm:Integral}. The series representation for
the off-diagonal entries $\pm\alpha D_{\pm}G_{1}$
can be obtained from Theorem~\ref{thm:Integral} 
and Propositions~\ref{prop:X} and \ref{prop:Xp}.

In this paragraph, we shall concentrate on the functions 
$D_{\pm}G_{1}$ and their particular values in the limit
$r\to0$, $\alpha\to0$, and $\beta\to0$, where as before,
$r=\abs{x}$, $x=(x_{1},x_{2},x_{3})\in\R^{3}$.

With the parameters as in Theorem~\ref{thm:Integral},
$\forall j=1,2,3$,
\[
\frac{\pd}{\pd x_{j}}G_{1}(x)=
\frac{1}{8\pi\sqrt{-\zeta}}\frac{\pd}{\pd x_{j}}
X^{\prime}\left(\frac{1}{2},\frac{3}{2};v\right)-
\frac{x_{j}}{8\pi r}X\left(\frac{3}{2},\frac{3}{2};u\right)
-\frac{r}{8\pi}\frac{\pd}{\pd x_{j}}
X\left(\frac{3}{2},\frac{3}{2};u\right)
\]
where the triplets $u=(u_{1},u_{2},u_{3})\in\C^{3}$,
$v=(v_{1},v_{2},v_{3})\in \C^{3}$ are given by 
\[
u=\left(\frac{\beta^{2}r^{4}}{64},
-\frac{\alpha^{2}r^{2}}{16},
-\frac{\zeta r^{2}}{4}\right),\quad 
v=\left(\frac{\beta^{2}}{4\zeta^{2}},
-\frac{\zeta \alpha^{2}}{\beta^{2}},
\frac{\zeta r^{2}}{4}\right).
\]
By (\ref{eq:pd3Xp}) in Proposition~\ref{prop:Xp},
\[
\frac{\pd}{\pd x_{j}}X^{\prime}\left(\frac{1}{2},\frac{3}{2};v\right)
=\frac{\pd v_{3}}{\pd x_{j}}
\frac{\pd}{\pd v_{3}}X^{\prime}\left(\frac{1}{2},\frac{3}{2};v\right)
=-\zeta x_{j}X^{\prime}\left(-\frac{1}{2},\frac{3}{2};v\right).
\]
By (\ref{eq:pd1X}) and (\ref{eq:pd3X}) in Proposition~\ref{prop:X},
\begin{align*}
\frac{\pd}{\pd x_{j}}X\left(\frac{3}{2},\frac{3}{2};u\right)
=&
\sum_{k=1}^{3}\frac{\pd u_{k}}{\pd x_{j}}
\frac{\pd}{\pd u_{k}}X\left(\frac{3}{2},\frac{3}{2};u\right) \\
=&
\frac{x_{j}}{2}\left(
\frac{\beta^{2}r^{2}}{45}X\left(\frac{7}{2},\frac{5}{2};u\right)
-\frac{2\zeta}{3}X\left(\frac{5}{2},\frac{3}{2};u\right)-
\frac{\alpha^{2}}{4}\pd_{2}X\left(\frac{3}{2},\frac{3}{2};u\right)
\right)
\end{align*}
where $\pd_{2}\equiv\pd/\pd u_{2}$ and,
by (\ref{eq:pd2X-0}),
\begin{equation}
\pd_{2}X\left(\frac{3}{2},\frac{3}{2};u\right)=
\sum_{n=1}^{\infty}\frac{ n u_{2}^{n-1} }{ 
\left(\frac{3}{2}\right)_{n}^{2} }\,
F_{1:1;0}^{0:0;0}\left(
\begin{array}{rr}
: & ;; \\
\left(\frac{3}{2}+n:2,1\right): & 
\left(\frac{3}{2}+n:1\right);;
\end{array}u_{1},u_{3}\right)
\label{eq:pd2X-1}
\end{equation}
and the series is absolutely convergent.

Combining all together,
\begin{align}
D_{\pm}G_{1}(x;\alpha,\beta,\zeta)=&
\left(\frac{\pd}{\pd x_{1}}\pm\im\frac{\pd}{\pd x_{2}}
\right)G_{1}(x;\alpha,\beta,\zeta) \nonumber \\
=&\frac{x_{1}\pm\im x_{2}}{8\pi}\left(
\sqrt{-\zeta}X^{\prime}\left(-\frac{1}{2},\frac{3}{2};
\frac{\beta^{2}}{4\zeta^{2}},
-\frac{\zeta \alpha^{2}}{\beta^{2}},
\frac{\zeta r^{2}}{4}\right)
\right. \nonumber \\
&-\frac{1}{r}X\left(\frac{3}{2},\frac{3}{2};
\frac{\beta^{2}r^{4}}{64},
-\frac{\alpha^{2}r^{2}}{16},
-\frac{\zeta r^{2}}{4}\right) \nonumber \\
&-\frac{r}{2}\left[
\frac{\beta^{2}r^{2}}{45}
X\left(\frac{7}{2},\frac{5}{2};
\frac{\beta^{2}r^{4}}{64},
-\frac{\alpha^{2}r^{2}}{16},
-\frac{\zeta r^{2}}{4}\right)
\right. \nonumber \\
&-\frac{2\zeta}{3}
X\left(\frac{5}{2},\frac{3}{2};
\frac{\beta^{2}r^{4}}{64},
-\frac{\alpha^{2}r^{2}}{16},
-\frac{\zeta r^{2}}{4}\right) \nonumber \\
&\left.\left.-\frac{\alpha^{2}}{4}
\pd_{2}X\left(\frac{3}{2},\frac{3}{2};
\frac{\beta^{2}r^{4}}{64},
-\frac{\alpha^{2}r^{2}}{16},
-\frac{\zeta r^{2}}{4}\right)
\right]\right)
\label{eq:DpmG1}
\end{align}
where $\pd_{2}X$ is given by (\ref{eq:pd2X-1});
the parameters are as in Lemmas~\ref{lem:X}--\ref{lem:Xp}.

1. Let $U_{\varepsilon}$ be an 
$\varepsilon$-neighborhood of the origin
$0\in\R^{3}$. By (\ref{eq:DpmG1}), 
\begin{equation}
D_{\pm}G_{1}(x;\alpha,\beta,\zeta)=
-\frac{\hat{x}_{1}\pm\im \hat{x}_{2}}{8\pi},\quad
\forall x\in U_{\varepsilon},\quad
\hat{x}_{j}=\frac{x_{j}}{r},\quad j=1,2
\label{eq:DpmG1origin}
\end{equation}
for $\varepsilon>0$ sufficiently small. Turns out that
the off-diagonal entries $\pm\alpha D_{\pm}G_{1}(x)$
of Green's function are well-defined $\forall x\in\R^{3}-\{0\}$.

2. By (\ref{eq:DpmG1}), in the limit $\alpha\to0$,
\begin{align}
D_{\pm}G_{1}(x;0,\beta,\zeta)=&
\frac{x_{1}\pm\im x_{2}}{8\pi}
\left(
\sqrt{-\zeta}\mrm{H}_{10}\left(-\frac{1}{2};\frac{3}{2};
\frac{\beta^{2}}{4\zeta^{2}},\frac{\zeta r^{2}}{4}\right)
\right. \nonumber \\
&-\frac{1}{r}F_{1:1;0}^{0:0;0}\left(
\begin{array}{rr}
: & ;; \\
\left(\frac{3}{2}:2,1\right): & 
\left(\frac{3}{2}:1\right);;
\end{array}\frac{\beta^{2}r^{4}}{64},
-\frac{\zeta r^{2}}{4}\right) \nonumber \\
&-\frac{r}{2}\left[
\frac{\beta^{2}r^{2}}{45}
F_{1:1;0}^{0:0;0}\left(
\begin{array}{rr}
: & ;; \\
\left(\frac{7}{2}:2,1\right): & 
\left(\frac{5}{2}:1\right);;
\end{array}\frac{\beta^{2}r^{4}}{64},
-\frac{\zeta r^{2}}{4}\right)
\right. \nonumber \\
&\left.\left.-\frac{2\zeta}{3}
F_{1:1;0}^{0:0;0}\left(
\begin{array}{rr}
: & ;; \\
\left(\frac{5}{2}:2,1\right): & 
\left(\frac{3}{2}:1\right);;
\end{array}\frac{\beta^{2}r^{4}}{64},
-\frac{\zeta r^{2}}{4}\right)\right]\right)
\label{eq:DpmG1alpha0}
\end{align}
$\forall\beta\geq0$, $\forall\zeta\in\C-
\left[-\beta,\infty\right)$, $\abs{\zeta}>\beta$,
$\forall x\in\R^{3}-\{0\}$. In view of 
\[
\mrm{H}_{10}\left(-\frac{1}{2};\frac{3}{2};
\zeta_{1},\zeta_{2}\right)=
\sum_{\sigma=\pm1}\frac{ \sigma\left(1+2\sigma 
\sqrt{\zeta_{1}} \right)^{\frac{3}{2}} }{ 6\sqrt{\zeta_{1}} } 
\,_{0}F_{1}\left(;\frac{5}{2};-\zeta_{2}
\left(1+2\sigma\sqrt{\zeta_{1}}\right)\right)
\]
$\forall\abs{\zeta_{1}}<\frac{1}{4}$, the additional
condition $\abs{\zeta}>\beta$ in (\ref{eq:DpmG1alpha0})
can be omitted. 
Note that $D_{\pm}G_{1}(x;0,\beta,\zeta)$
is also easy to obtain from (\ref{eq:G1alpha0-2}):
\begin{equation}
D_{\pm}G_{1}(x;0,\beta,\zeta)=
\frac{x_{1}\pm\im x_{2}}{8\pi\beta r^{2}}
\left(
\left(\sqrt{\beta-\zeta}+\frac{1}{r}\right)
\e^{ -r\sqrt{\beta-\zeta} }
-\left(\sqrt{-\beta-\zeta}+\frac{1}{r}\right)
\e^{ -r\sqrt{-\beta-\zeta} }\right)
\label{eq:DpmG1alpha0-1}
\end{equation}
$\forall\beta\geq0$, $\forall\zeta\in\C-
\left[-\beta,\infty\right)$, $\forall x\in\R^{3}-\{0\}$.

By (\ref{eq:DpmG1alpha0})--(\ref{eq:DpmG1alpha0-1}),
the Green's function $\mathcal{G}_{R}(x;0,\beta,\zeta)$ is diagonal.

3. In the limit $\beta\to0$,
$\forall\alpha\geq0$, $\forall\zeta\in\C-\left[
\frac{1}{4}\alpha^{2},\infty\right)$, $\abs{\zeta}>
\frac{1}{4}\alpha^{2}$, $\forall x\in\R^{3}-\{0\}$,
\begin{align}
D_{\pm}G_{1}(x;\alpha,0,\zeta)=&
\frac{x_{1}\pm\im x_{2}}{8\pi}
\left(
\sqrt{-\zeta}\mrm{H}_{3}\left(-\frac{1}{2},1;\frac{3}{2};
-\frac{\alpha^{2}}{4\zeta},\frac{\zeta r^{2}}{4}\right)
\right. \nonumber \\
&-\frac{1}{r}F_{1:1;0}^{0:1;0}\left(
\begin{array}{rr}
: & 1;; \\
\frac{3}{2}: & 
\frac{3}{2};;
\end{array}-\frac{\alpha^{2}r^{2}}{16},
-\frac{\zeta r^{2}}{4}\right) \nonumber \\
&+\frac{r}{2}\left[
\frac{2\zeta}{3}
F_{1:1;0}^{0:1;0}\left(
\begin{array}{rr}
: & 1;; \\
\frac{5}{2}: & 
\frac{3}{2};;
\end{array}-\frac{\alpha^{2}r^{2}}{16},
-\frac{\zeta r^{2}}{4}\right) \right. \nonumber \\
&\left.\left.+\frac{\alpha^{2}}{9}
F_{1:1;0}^{0:1;0}\left(
\begin{array}{rr}
: & 2;; \\
\frac{5}{2}: & 
\frac{5}{2};;
\end{array}-\frac{\alpha^{2}r^{2}}{16},
-\frac{\zeta r^{2}}{4}\right)\right]\right).
\label{eq:DpmG1beta0}
\end{align}
Taking both $\alpha=\beta=0$, one finds from
(\ref{eq:DpmG1alpha0})--(\ref{eq:DpmG1beta0}) that
\[
D_{\pm}G_{1}(x;0,0,\zeta)=
-\frac{\hat{x}_{1}\pm\im \hat{x}_{2}}{8\pi}
\e^{-r\sqrt{-\zeta}}
\]
$\forall\zeta\in\C-\left[0,\infty\right)$,
$\forall x\in\R^{3}-\{0\}$. 
\section{Discussion}\label{sec:8}
The series representation of Green's function is 
well-suited for further spectral analysis of $H$ (\ref{eq:Hamilton}). 
Suppose that $H_{0}$ is the operator $H$ restricted to the set of
compactly supported smooth functions that vanish at the origin.
The operator $H_{0}$ is symmetric but not self-adjoint. Self-adjoint
extensions of $H_{0}$, say $\tilde{H}$, can be found by applying the 
singular perturbation theory \cite{Albeverio00}. From this point
of view, $H$ is a trivial extension of $H_{0}$. The extensions incorporate 
the operators that are usually referred to as the Hamiltonians with
point-interaction. In physical applications, for example in ultracold
atomic gases, operator $\tilde{H}$ would describe the 
spin-orbit coupled Hamiltonian considered in the presence of magnetic
field with the impurity scattering treated via the zero-range interaction.
Without spin-orbit coupling, that is, for $\alpha=0$, self-adjoint
extensions and their spectral properties have been examined in
\cite{Cacciapuoti09}; in this case, the ability to obtain exact eigenvalues 
is clear from a simple structure of Green's function, see (\ref{eq:G1alpha0-2}) 
and (\ref{eq:G2alpha0-2}). For a general $\alpha\geq0$, however,
it is convenient to represent the resolvent of $\tilde{H}$ in terms of Krein's 
$Q$-matrix function. One can show that,
since $D_{\pm}G_{1}(x;\alpha,\beta,\zeta)$ is independent
of $\zeta$ (\ref{eq:zeta}) in the neighborhood of the origin
$0\in\R^{3}$ (\ref{eq:DpmG1origin}), 
the $Q$-matrix function is made up of only
$G_{1}(0;\alpha,\beta,\zeta)$ (\ref{eq:Spec1}) and 
$G_{2}^{\mrm{ren}}(0;\alpha,\beta,\zeta)$
(\ref{eq:Spec2}). The analysis of these functions 
leads to the transcendental equation
with respect to the eigenvalue of $\tilde{H}$. 
The results subsequent to the algebraic treatment of the present discussion
will be announced elsewhere.
\section*{Acknowledgement}
It is a pleasure to thank the anonymous referees for
their valuable comments and critical remarks which have
helped to improve the present version of the manuscript.
The research was partially funded by 
the European Social Fund under the Global Grant measure.
\bibliographystyle{elsarticle-num}

\end{document}